\documentclass[12pt]{article}
\usepackage{amsmath,amsfonts,amssymb}
\usepackage{enumerate}
\usepackage[dvipsnames]{xcolor}
\usepackage{geometry}
\usepackage[parfill]{parskip}
\usepackage{graphicx}
\usepackage{algorithm,algpseudocode}
\usepackage{subcaption}
\usepackage{natbib}
\usepackage{amsthm}
\usepackage[small]{titlesec}
\usepackage{authblk}
\graphicspath{{./Images/}}
\DeclareUnicodeCharacter{2032}{'}

\newtheorem{theorem}{Theorem}
\newtheorem{lemma}{Lemma}
\newtheorem{corollary}{Corollary}

\newtheoremstyle{break}
  {\topsep}{\topsep}
  {\upshape}{}
  {\bfseries}{}
  {\newline}{}
\theoremstyle{break}

\begin{document}
\title{Kernel Density Balancing}

 \author[1]{John Park}
  \author[1,2,*]{Ning Hao}
  \author[1,2]{Yue Selena Niu}
  \author[3]{Ming Hu}
  \affil[1]{Department of Mathematics, the University of Arizona} 
  \affil[2]{Statistics \& Data Science GIDP, the University of Arizona}
  \affil[3]{Department of Quantitative Health Sciences, Lerner Research Institute, Cleveland Clinic Foundation}
  \affil[*]{Corresponding author: nhao@arizona.edu}
  
\maketitle

\begin{abstract}
High-throughput chromatin conformation capture (Hi-C) data provide insights into the 3D structure of chromosomes, with normalization being a crucial pre-processing step. A common technique for normalization is matrix balancing, which rescales rows and columns of a Hi-C matrix to equalize their sums. Despite its popularity and convenience, matrix balancing lacks statistical justification. In this paper, we introduce a statistical model to analyze matrix balancing methods and propose a kernel-based estimator that leverages spatial structure. Under mild assumptions, we demonstrate that the kernel-based method is consistent, converges faster, and is more robust to data sparsity compared to existing approaches.

\end{abstract}


\section{Introduction}
Matrix balancing is a widely used technique for removing biases, correcting measurement errors, and enabling meaningful comparisons in matrix data across fields such as probability theory, economics, and biology \citep{Idel16}. Informally, it adjusts the entries of a nonnegative matrix so that each row and column sums to specified target values. This method has practical applications in many areas. For example, it is used in economics to balance input-output and migration models \citep{Bacharach65, ULT66, deBeer2010}, in genomics to normalize measurements of DNA interactions in cells \citep{Lyu20,Imakaev12}, in Markov chain modeling to establish valid transition probabilities \citep{Sinkhorn64}, and in network analysis to fairly evaluate connections or interactions among nodes \citep{Krupp79}.  

Often, the target row and column sums are all ones, in which case the matrix is called doubly stochastic. Matrix balancing adjusts a square matrix into a doubly stochastic form by multiplying it on the left and right by positive diagonal matrices. Under mild conditions, it has been established that a unique solution exists, although no closed-form expression is generally available \citep{Sinkhorn64}. Primitive matrix balancing algorithms typically rely on iterative procedures that alternately adjust rows and columns until convergence. One notable algorithm is the Sinkhorn-Knopp algorithm, a popular iterative method alternating between row and column normalizations to reach convergence \citep{Sinkhorn67}. These methods remain popular due to their simplicity and ease of implementation, and have been extended to continuous settings as shown in \cite{BLN94}. In the symmetric case, it is required that both input and output matrices are symmetric. The symmetric analogues of these algorithms have also been developed and are commonly applied. The existence and uniqueness of the doubly stochastic solution, as well as the convergence of iterative balancing methods, have been well studied for square matrices with positive entries \citep{Sinkhorn67,soules1991rate}. However, scaling matrix balancing to large matrices remains an active area of research, with continued efforts to design more efficient algorithms and computational techniques \citep{knight2008sinkhorn,johnson2009scaling,KR12,allen2017much}. For a comprehensive review and recent developments in matrix balancing and its applications, see \cite{Idel16}.

One of the most prominent modern applications of matrix balancing is in the normalization of Hi‑C data, which are generated from high-throughput chromatin conformation capture experiments designed to study the three-dimensional organization of the genome. Specifically, Hi-C techniques measure how frequently pairs of genomic regions interact, producing large matrices representing interaction frequencies. Due to technical biases such as uneven distribution of restriction enzyme cut sites, guanine-cytosine (GC) content, and differences in mappability, raw Hi-C matrices require balancing before analysis. Matrix balancing algorithms have been widely applied to remove these biases, ensuring accurate comparisons across genomic regions \citep{Rao14,Imakaev12,Lyu20}.  

Despite its widespread use and proven utility, current matrix balancing approaches face several challenges. First, the existing algorithms are sensitive to data sparsity, a common characteristic of high-resolution Hi-C data, often resulting in convergence issues or unstable normalization. Second, certain applications of matrix balancing involve data matrices with inherent spatial structure. For example, in Hi-C data, the rows and columns are ordered by the genomic locations along the chromosome. Standard matrix balancing methods generally do not account for this spatial dependency, leading to information loss. Third, the large size of high-resolution Hi-C data often necessitates entry-wise storage, requiring conversion to matrix form before processing, which could be intractable. Finally, most existing works treat normalization as a deterministic process, overlooking data uncertainty and lacking a robust statistical framework to assess performance.
 
In this paper, we propose a novel nonparametric model that treats the observed matrix for balancing as a pre-binned random sample from an underlying distribution. This perspective enables us to quantify data uncertainty and establish the statistical consistency of matrix balancing algorithms. Within this framework, we study a continuous analogue of matrix balancing, referred to as the density balancing problem. We introduce a kernel smoothing–based algorithm for density balancing and demonstrate that it achieves a faster convergence rate than traditional matrix balancing methods, which can be understood as a histogram-based balancing algorithm. Theoretically, we establish convergence rate results of density balancing algorithms, providing stronger justification for their use in processing Hi-C data. Methodologically, kernel balancing leverages spatial structure by borrowing information from neighboring entries through a data-adaptive bandwidth, effectively mitigating sparsity and smoothing over zero entries. Computationally, the algorithm also exploits the structure of raw Hi-C data, allowing it to operate directly on the sparse, entry-wise format without requiring full matrix reconstruction. This makes it feasible to apply the method at the highest available resolution, even when the matrix is too large to store explicitly. Together, these contributions address a key theoretical gap in matrix balancing and offer a practical solution for high-resolution Hi-C data analysis.

\section{Matrix Balancing}
\label{Section::Background}
\subsection{Background}\label{Section::2.1}
Let $C$ be a positive $n\times n$ matrix and $e$ represent the $n\times 1$ vector consisting of all $1$s. The goal of matrix balancing is to find a matrix $P$ and diagonal matrices $D_1$ and $D_2$ such that 
\[P=D_1CD_2,\quad Pe=e,\quad \text{and}\quad P^{\top}e=e.\] 
Such a matrix $P$ is called doubly stochastic, the matrices $D_1$ and $D_2$ are referred to as balancing matrices, and their diagonal entries balancing vectors. 

Matrix balancing arises naturally in various applications where equilibration of marginal quantities is important. For example, in Markov Chain estimation, a row-normalized observed transition matrix is often used as an estimate for the true transition matrix. If the transition matrix is also known to be doubly stochastic, matrix balancing may be used to turn an arbitrary observed transition matrix into a doubly stochastic matrix \citep{Sinkhorn64}. In population analysis, immigration and emigration data are often summarized in matrices, with rows and columns indexed by countries and entries representing population flows between them. Due to reporting inconsistencies or missing data, row and column sums often mismatch, and balancing techniques can adjust them to match specified targets \citep{deBeer2010}. In genomics, particularly in the analysis of Hi-C data, balancing methods are used to correct systematic biases in measuring chromatin interactions, improving the reliability of downstream interpretations \citep{Lyu20}. Similar needs arise in other domains where flows must be balanced, including trade models in economics, traffic flows in networks, and optimal transport.

Because the matrix balancing problem arises naturally in a variety of contexts, related algorithms have been studied and named as early as the 1940s. A comprehensive historical review is provided in \citet{Idel16}. A foundational result is Sinkhorn’s Theorem \citep{Sinkhorn64}, which states that if $C$ is strictly positive, then there exists a unique doubly stochastic matrix $P$ of the form $D_1 C D_2$. The associated Sinkhorn-Knopp (SK) algorithm, outlined in Algorithm~\ref{Alg::SK}, alternates row and column normalizations until convergence. Formal statements of the relevant theorems are given in Appendix~\ref{Section::KnownThms}.

\begin{algorithm}
\caption{Sinkhorn-Knopp (SK)}\label{Alg::SK}
\begin{algorithmic}[1]
    \Require Positive matrix $C$; tolerance $\varepsilon$
    \Ensure Balanced matrix $P$ and balancing matrices $D_1$ and $D_2$
    \State $R\gets \text{diag}\left(Ce\right)$, $S\gets \text{diag}\left(C^{\top}e\right)$ 
    \State $D_1^{(0)}\gets R^{-1}$, $D_2^{(0)}\gets I$, $P^{(0)}\gets C$, $t\gets 0$
    \While{$\left\lVert R-I\right\rVert_{\infty}>\varepsilon$\textbf{ or }$\left\lVert S-I\right\rVert_{\infty}>\varepsilon$}
        \State $\widetilde{P}^{(t+1)}\gets R^{-1}P^{(t)}$\Comment{$\widetilde{P}^{(t+1)}$ is the row-normalized $P^{(t)}$}
        \State $S\gets \text{diag}\left( [\tilde{P}^{(t+1)}]^{\top}e\right)$
        \State $P^{(t+1)}\gets \widetilde{P}^{(t+1)}S^{-1}$\Comment{${P}^{(t+1)}$ is the column-normalized $\widetilde{P}^{(t+1)}$}
        \State $R\gets \text{diag}\left( P^{(t+1)}e\right)$
        \State $D^{(t+1)}_1\gets R^{-1}D^{(t)}_1$, $D^{(t+1)}_2\gets D^{(t)}_2S^{-1}$, $t\gets t+1$
    \EndWhile
    \State \Return $P,D_1,D_2$
\end{algorithmic}
\end{algorithm}

A similar problem may be asked for scaling matrices to arbitrary row and column sums. Let $r$ and $c$ be positive $n\times1$ vectors such that $r^{\top}e=c^{\top}e$. The goal of matrix scaling is to find a matrix $P$ and diagonal matrices $D_1$ and $D_2$ such that $P=D_1CD_2$, $Pe=r$ and $P^{\top}e=c$. \cite{Sinkhorn672} showed that having a strictly positive matrix $C$ is also a sufficient condition for the existence of a unique solution to this problem. Similarly, a simple algorithm of alternating row and column scaling converges to the solution; see Theorem~\ref{Th::ASK} in Appendix~\ref{Section::KnownThms} for details.

When the matrix $C$ is symmetric, the row and column sums are identical. The goal of matrix balancing can then be summarized by a single equation --- finding a diagonal matrix $D$ such that 
\begin{equation*}
    DCDe=e.
\end{equation*}
The matrix $D$ is then the unique balancing matrix and its diagonal entries the unique balancing vector. 
In the symmetric setting, Algorithm \ref{Alg::SK} is not able to preserve symmetry during the iteration process. A variant, shown in Algorithm~\ref{Alg::SSK}, addresses this issue by scaling the rows and columns simultaneously in each step, thereby maintaining symmetry \citep{KnightRuiz14}. 

\begin{algorithm}
\caption{Symmetric Sinkhorn-Knopp (SSK)}\label{Alg::SSK}
\begin{algorithmic}[1]
    \Require Symmetric positive matrix $C$; tolerance $\varepsilon$
    \Ensure Balanced matrix $P$ and balancing matrix $D$
    \State $R\gets \text{diag}\left({Ce}\right)$
    \State $D^{(0)}\gets R^{-\frac12}$, $P^{(0)}\gets C$, $t\gets0$
    \While{$\left\lVert R-I\right\rVert_{\infty}>\varepsilon$} 
        \State $P^{(t+1)}\gets R^{-\frac12}P^{(t)}R^{-\frac12}$ \label{AlgStep::SSKStart}
        \State $R\gets\text{diag}\left(P^{(t)}e\right)$ \label{AlgStep::SSKEnd}
        \State $D^{(t+1)}\gets D^{(t)}R^{-\frac12}$, $t\gets t+1$
    \EndWhile
    \State \Return $P,D$
\end{algorithmic}
\end{algorithm}

Researchers have also studied the continuous analogue of this problem. Let $f(x,y)$ be a bounded, positive, continuous function defined on the unit square. Given two positive continuous functions on the unit interval, $r(x)$ and $c(y)$, the goal of density balancing is to find functions $h_1$ and $h_2$ on the unit interval such that 
\begin{align*}
	\int_0^1h_1(x)f(x,y)h_2(y)dx&=r(y)\text{ for all }y\in[0,1]\\
	\int_0^1h_1(x)f(x,y)h_2(y)dy&=c(x)\text{ for all }x\in[0,1].
\end{align*}
\cite{BLN94} proved existence and uniqueness results to the kernel scaling problem under a similar positivity condition and a convergence result for a continuous version of  Algorithm \ref{Alg::SK}; see Theorem \ref{Th::KSK} in the appendix. These results are generalizations of the case where $f(x,y)$ is symmetric and $r(y)=c(x)$, for which similar results were shown earlier \citep{Nowosad64}.  Just as in the discrete case, a continuous version of Algorithm \ref{Alg::SSK} can be constructed, with a continuous function replacing the matrix and integrals replacing the row sums. 

\begin{algorithm}
\caption{Continuous Symmetric Sinkhorn-Knopp (CSSK)}\label{Alg::CSSK}
\begin{algorithmic}[1]
    \Require Symmetric density $f$ on $[0,1]^2$; tolerance $\varepsilon$
    \Ensure Balanced density $p(x,y)$ and balancing function $h$
    \State $r(y) \gets \int_0^1f(x,y)dx$ 
    \State $h(y)\gets \frac{1}{\sqrt{r(y)}}$, $p^{(0)}(x,y)\gets f(x,y)$, $t\gets0$
    \While{$\left\lVert r-1\right\rVert_{\infty}>\varepsilon$} 
        \State $p^{(t+1)}(x,y)\gets \frac{p^{(t)}(x,y)}{\sqrt{r(x)r(y)}}$ 
        \State $r(y)\gets\int_0^1f(x,y)dx$ 
        \State $h(y)^{(t+1)}\gets \frac{h(y)}{\sqrt{r(y)}}$, $t\gets t+1$
    \EndWhile
    \State \Return $p,h$
\end{algorithmic}
\end{algorithm}

For the remainder of the paper, we focus on the symmetric case, although most of our results can be extended to the general setting.

\subsection{Applications to Hi-C data}
The output of a Hi-C experiment consists of fragments of interacting pieces of DNA. Analyzing these fragments directly is noisy, and the observed interactions are instead binned into a fixed genomic interval. The resulting matrices are high-dimensional, sparse, and symmetric, and their entries represent the observed interaction frequency between two genomic locations. The size of the bins can be adjusted depending on the goals of the analysis, resulting in square matrices with dimensions ranging from $250$ to $250000$. Due to their size, Hi-C matrices are typically stored in coordinate list format, as tuples of row index, column index, and nonzero count value.

Following this pre-processing step, normalization is performed on the Hi-C matrix to correct systematic biases. One class of normalization methods is based on the ``equal loci visibility'' assumption, which states that the number of interactions should be roughly the same for each region of DNA.  Mathematically, this corresponds to the assumption that the Hi-C matrix should have constant row and column sums. These methods also assume that the biases are factorizable, that is, the observed frequency satisfies $C_{ij}=d_id_jP_{ij}$ where $d$ is a vector of biases and $P$ is the true contact frequency \citep{Imakaev12}. After scaling the problem appropriately, this corresponds to finding a diagonal matrix $D$ such that $DCD=P$ for a doubly stochastic matrix $P$, which is exactly the symmetric matrix balancing problem introduced in the previous subsection. 

The matrix balancing techniques introduced in Section \ref{Section::2.1} are guaranteed to converge when the input matrix is strictly positive. However, the presence of zeros in Hi-C data poses practical challenges. In particular, high-resolution Hi-C matrices are extremely sparse, sometimes requiring alternative normalization techniques.

\section{Statistical Model}
\label{Section::SM}
In most applications and analyses, matrix balancing is treated as a deterministic problem, overlooking the uncertainty inherent in the data generation process. More importantly, existing methods often fail to fully utilize the spatial relationships among rows (and columns), leading to information loss. Consequently, many matrix balancing algorithms are statistically inefficient, and their statistical properties are not fully understood.    

In this section, we address these limitations by introducing a nonparametric statistical model that treats matrix entries as binned random observations drawn from a continuous joint density. Under this framework, the data matrix corresponds to a two-dimensional histogram estimator of the underlying density, whose statistical properties can be rigorously analyzed. Moreover, this model allows the development of new balancing techniques based on kernel density estimation.

Let $f^*(x,y)$ be a joint density on the unit square $[0,1]^2\subset\mathbb{R}^2$ satisfying the following two conditions:
\begin{enumerate}[(a)]
	\item Symmetry: $f^*(x,y)=f^*(y,x)$. 
	\item Marginal Uniformity: $f^*(x,y)$ has standard uniform marginals, \[f_x^*(x)=\int_0^1f(x,y)dy=1\quad\text{and}\quad  f_y^*(y)=\int_0^1f(x,y)dx=1.\]
\end{enumerate}

Densities that satisfy these conditions are sometimes referred to as symmetric copulas in the statistical literature \citep{Sklar59}. The marginals of a copula typically represent the cumulative distribution functions (CDFs) of the variables of interest, with research focusing on understanding their dependence structure.
To highlight the connection between our research question and the matrix balancing problem, we refer to such a density function as a symmetric and doubly stochastic density (SDSD).

In Hi-C data analysis, we may assume that, in the absence of experimental bias, the data are independently drawn from an SDSD, then binned into a matrix. In practice, however, experimental bias distorts the data, so that the observed samples are drawn from a symmetric density $f(x, y)$, derived from an underlying SDSD $f^*(x, y)$ through a factorizable bias function
\begin{equation}\label{Eq::Cbal}
f(x,y) = g(x)f^*(x,y)g(y),
\end{equation}
where $g$ is a positive, continuous function. While the resulting density $f$ remains symmetric, it generally does not satisfy the marginal uniformity condition required by an SDSD.

Under this framework, we aim to estimate the underlying SDSD $f^*$ and the bias function $g$, given a random sample from the distorted density $f$. When $f$ is known, $f^*$ can be recovered from $f$ by iteratively normalizing the marginals of $f$ as we mentioned in Section \ref{Section::Background}. In practice, $f$ is unknown and estimated by either a histogram or a kernel density estimator, which leads to estimators of $f^*$ via balancing algorithms. Before studying these approaches and their theoretical properties, we briefly demonstrate the Hi-C data under our framework in detail.   
 
The raw Hi-C data are usually provided in a three-column format, where each row records a pair of genomic loci, represented by their horizontal and vertical positions, and the corresponding contact frequency. This format reflects the inherently sparse nature of Hi-C data, especially at high resolutions, where most locus pairs have zero or low interaction counts. The data can be interpreted as binned observations from a continuous interaction surface, with symmetry implied by the nature of chromatin contact (i.e., contact between locus $x$ and $y$ is equivalent to that between $y$ and $x$). To perform analysis, these data are often aggregated into matrices by binning the genome into fixed-length intervals, but the initial three-column representation preserves full resolution and is more computationally efficient for storage and preprocessing.  

Hypothetically, the Hi-C data form a random sample from a density $f$. To make it precise, we scale the chromosome of interest to the unit interval $[0,1]$ and assume that observed are $N$ independent and identically distributed (IID) pairs $\left\{\left(X_k,Y_k)\right)\right\}_{k=1}^N$ from $f$ on the unit square. We may always assume $X_k\leq Y_k$ as the order does not matter. At a specific resolution, the chromosome is divided into $n$ equal-length intervals, the $n\times n$ Hi-C matrix $C=\left(C_{ij}\right)$ is obtained as 
\begin{equation*}
C_{ij}=C_{ji}=\left|\left\{k:\left(X_k,Y_k\right)\in\left[\frac{i-1}{n},\frac{i}{n}\right)\times\left[\frac{j-1}{n},\frac{j}{n}\right)\right\}\right|,\quad 1\leq i< j\leq n.
\end{equation*}
This can be roughly understood as there are $C_{ij}$ observations in a small square neighborhood centered at   \[\left(\frac{2i-1}{2n},\frac{2j-1}{2n}\right).\]
We should point out that in practice, the Hi-C data are given as a pre-binned form
\[\left\{\left(\frac{2i-1}{2n},\frac{2j-1}{2n},C_{ij}\right):1\leq i< j\leq n,C_{ij}> 0\right\}.\]
Nevertheless, conceptually, it is helpful to understand the balancing algorithms using the random observations $\left\{\left(X_k,Y_k)\right)\right\}_{k=1}^N$. Practically, by using the highest resolution, the information loss on the precise loci is negligible. Technically, we will use a triple $(X_k,Y_k,c_k)$ to take into account the multiplicity in the kernel density estimator.

\section{Balancing Algorithms}
\label{Section::Alg}
\subsection{Kernel Density Balancing}
Under the model described in the previous section, the input matrix to a balancing algorithm can be interpreted as a histogram estimator of the density function $f$. Consequently, the resulting balancing vector may be viewed as a discrete estimator of $g^{-1}$. From a density estimation perspective, it is well known that histogram estimators do not achieve an optimal convergence rate and are outclassed by kernel density estimators \citep{wand1994kernel}. To develop a statistically optimal procedure, one may consider a kernel smoothing-based approach. For example, a conceptually simple approach is to apply a continuous analogue of the matrix balancing algorithm to a 2D kernel density estimate of $f$, yielding an estimate of $f^*$. Nevertheless, it is not convenient to implement a 2D kernel density estimator and execute a balancing algorithm on it. Specifically, a naive implementation would require repeated numerical integration of a 2D function as shown in Algorithm \ref{Alg::CSSK}. 

We propose to use a symmetric product kernel $K_2$ for density estimation, bypassing the 2D density estimation step and significantly simplifying the implementation. To elaborate, we observe that the marginal of a 2D density estimator is itself a density estimator of the marginal distribution, as formally stated in Lemma 1 below. Using this trick allows us to conduct the balancing algorithm from the raw Hi-C data directly. Moreover, in the iterations, instead of recording marginals, we assign and update a weight for each observation to track the balancing procedure, as shown in Algorithm \ref{Alg::KSK}.

\begin{lemma}\label{lemma1}
    Let $\{(X_k,Y_k)\}_{k=1}^N$ be a random sample from a density $f$ defined on $[0,1]^2$, and $K_2(x,y)=K(x)K(y)$ be a symmetric product kernel. Then the kernel density estimators satisfy 
    \begin{equation}\label{equa2}
        \hat f_x(x)=\int_{-\infty}^\infty\hat f(x,y)dy,
    \end{equation} 
    where
    \[\hat f(x,y)=\frac{1}{Nh^2}\sum_{k=1}^N K_2\left(\frac{x-X_k}{h},\frac{y-Y_k}{h}\right), \text{ and }\hat f_x(x)=\frac{1}{Nh}\sum_{k=1}^N K\left(\frac{x-X_k}{h}\right).\]
    Let $\{(X_k,Y_k)\}_{k=1}^N$ be a random sample from a symmetric density $f$. Then Equation \eqref{equa2} holds with symmetric density estimators
    \begin{equation*} 
     \hat f(x,y)=\frac{1}{2Nh^2}\sum_{k=1}^N \left[K_2\left(\frac{x-X_k}{h},\frac{y-Y_k}{h}\right)+K_2\left(\frac{x-Y_k}{h},\frac{y-X_k}{h}\right)\right],   
    \end{equation*}
     and 
    \[\hat f_x(x)=\frac{1}{2Nh}\sum_{k=1}^N \left[ K\left(\frac{x-X_k}{h}\right)+K\left(\frac{x-Y_k}{h}\right)\right].\]
\end{lemma}

\begin{algorithm}
\caption{Kernel Sinkhorn-Knopp (KSK)}\label{Alg::KSK}
\begin{algorithmic}[1]
    \Require Observations $S=\left\{\left(X_k,Y_k,c_k\right)\right\}_{k=1}^N$; tolerance $\varepsilon$
    \Ensure Balancing function $g$.
    \State $g^{(0)}\left(x\right)\gets1$, $w^{(0)}_k\gets1$
    \State Use Equation \eqref{Eq::KSKInMarg} to compute a kernel density estimate $r^{(0)}(x)$
    \While{$\sup_{0\leq x\leq1}\left\{\left\lvert r(x)-1\right\rvert\right\}<\varepsilon$}
        \State $g^{(t+1)}\left(x\right)\gets g^{(t)}(x)r^{(t)}(x)$
        \State $w^{(t+1)}_k\gets1/\sqrt{g\left(X_k\right)g\left(Y_k\right)}$
        \State Use Equation \eqref{Eq::KSKUpMarg} to compute a new kernel density estimate $r^{(t+1)}(x)$
    \EndWhile
    \State \Return$g(x)$
\end{algorithmic}
\end{algorithm}

Algorithm \ref{Alg::KSK} can be applied to a pre-binned sample $\{(X_k,Y_k,c_k)\}_{k=1}^N$ or a regular sample $\{(X_k,Y_k)\}_{k=1}^N$ as a special case with $c_k=1$ for all $k$.
For a pre-binned sample $\{(X_k,Y_k,c_k)\}_{k=1}^N$ from a symmetric density, the $2$D kernel density estimate is
 \begin{equation}\label{equa3}
     \hat f(x,y)=\frac{1}{2Nh^2}\sum_{k=1}^N c_k\left[K_2\left(\frac{x-X_k}{h},\frac{y-Y_k}{h}\right)+K_2\left(\frac{x-Y_k}{h},\frac{y-X_k}{h}\right)\right].   
    \end{equation}
By Lemma \ref{lemma1}, its marginal is the $1$D kernel density estimate given below 
\begin{equation}
\label{Eq::KSKInMarg}
	r(x)=\frac{1}{2Nh}\sum_{k=1}^N c_{k} \left[K\left(\frac{x-X_k}{h}\right)+K\left(\frac{x-Y_k}{h}\right)\right].
\end{equation}

Following the strategy of Algorithm \ref{Alg::CSSK}, in the first iteration of a balancing procedure, we would calculate
\[f^{(1)}(x,y)=\hat f(x,y)/\sqrt{r(x)r(y)},\]
and evaluate its marginal density by integration. This is equivalent to assigning a weight $w_k=1/\sqrt{r(X_k)r(Y_k)}$ for each observation, where $f^{(1)}$ is the new density estimate under the weight assignment, and its marginal density is formulated as

\begin{equation}
\label{Eq::KSKUpMarg}
	r(x)=\frac{1}{2N}\sum_{k=1}^N\frac{c_kw_{k}}{h}\left[K\left(\frac{x-X_k}{h}\right)+K\left(\frac{x-Y_k}{h}\right)\right].
\end{equation}
By continuing the iterations, we arrive at Algorithm \ref{Alg::KSK}, which requires only one-dimensional kernel density estimation with updated weight assignments.

\subsection{Bandwidth Selection}
Bin width or bandwidth selection is crucial in density estimation. We propose a two-fold cross-validation (CV) for bandwidth selection in Algorithm $\ref{Alg::KSK}$, illustrated below. 

\begin{algorithm}
\caption{Bandwidth Selection}\label{Alg::KSKCV} 
\begin{algorithmic}[1]
    \Require A random sample $S=\left\{\left(X_k,Y_k)\right)\right\}_{k=1}^N$; a set $\mathcal{H}$ of candidate bandwidths 
    \Ensure Selected bandwidth for running Algorithm \ref{Alg::KSK}  
    \State Randomly split $S$ into two disjoint subset $S_1$ and $S_2$
    \State For each bandwidth $h\in\mathcal{H},$ run Algorithm \ref{Alg::KSK} to obtain $g_1$ and $g_2$, respectively
    \State Apply $g_1$ to $S_2$ and $g_2$ to $S_1$ and compute an evaluation score for fold.
    \State Select the bandwidth with the best average evaluation score 
\end{algorithmic}
\end{algorithm}

In Step 3 of Algorithm \ref{Alg::KSKCV}, we evaluate how well the bias function $g_1$ learned from $S_1$ balances $S_2$, and vice versa. From the perspective of Algorithm \ref{Alg::KSK}, applying $g_1$ to balance $S_2$ is equivalent to assigning a weight $w_k = 1/\sqrt{g(X_k)g(Y_k)}$ to each observation $(X_k, Y_k) \in S_2$. Evaluating the balancing performance then amounts to comparing the marginal distributions of the weighted sample against the standard uniform distribution. Therefore, we compute the empirical CDF of the weighted sample for comparison. While various discrepancy measures could be used, we adopt the Cram\'{e}r-von Mises criterion \citep{anderson1952asymptotic} to assess the deviation between the empirical CDF and the CDF of the uniform distribution.
When conducting the CV procedure, if the data are  pre-binned, as is the case for Hi-C data, the bin counts are treated as repeated observations, and the data are split accordingly.

The same strategy can be applied for bin width selection in histogram-based methods. We illustrate the CV procedure in Algorithm \ref{Alg::SSKCV}. For each possible number of bins $B$, the data are randomly split and binned to create $B\times B$ matrices. For each fold, an evaluation score is computed to determine the quality of the estimator. While several reasonable evaluation scores are possible, we choose cosine similarity due to its computational simplicity. 

\begin{algorithm}
\caption{Bin Size Selection}\label{Alg::SSKCV}
\begin{algorithmic}[1]
   \Require Observations $S=\left\{\left(X_k,Y_k\right)\right\}_{k=1}^N$; A set $\mathcal{B}\subset\mathbb{N}$ of candidate numbers of bins in $[0,1]$
   \Ensure Bin size for Algorithm \ref{Alg::SSK}.
   \State Randomly split $S$ into two disjoint subsets $S_1$ and $S_2$
   \For{each $B\in\mathcal{B}$}
       \State Split the interval $[0,1]$ into $B$ equally-sized sub-intervals 
       \State Bin $S_1$ and $S_2$ to generate two matrices $C_1$ and $C_2$
       \State Use Algorithm \ref{Alg::SSK} on $C_1$ and $C_2$ to obtain $D_1$ and $D_2$, respectively. 
       \State $r_1\gets D_1C_2D_1e$, $r_2\gets D_2C_1D_2e$. 
       \State Compute a score for each marginal, $s_i\gets\frac{r_i\cdot e}{\left\lVert r_i\right\rVert\left\lVert e\right\rVert}$
       \State $s\gets \frac{s_1+s_2}{2}$
   \EndFor
   \State Select the bin size with the best average evaluation score
\end{algorithmic}
\end{algorithm}

\section{Theoretical Results}
\label{Section::TR}
In this section, we view the inputs of Algorithms \ref{Alg::SSK} and \ref{Alg::KSK} as estimators to the distorted density function and discuss the statistical properties of the outputs of balancing algorithms as estimators to the underlying SDSD. We first note that both algorithms do converge. With Theorems \ref{Th::SK} and \ref{Th::KSK}, this guarantees that their outputs are the unique solutions to the corresponding balancing problems. We note that Algorithm \ref{Alg::SSK} has been previously studied and its convergence shown \citep{KnightRuiz14}. However, we present an alternate proof that is more easily generalized to the continuous case and present it for completeness.

\begin{theorem}\label{Th::CSSKConv}
    Let $f$ be a bounded, continuous, strictly positive, symmetric density on $[0,1]^2$, with unique functions $f^*$ and $g$ satisfying Equation \eqref{Eq::Cbal}. Algorithm \ref{Alg::CSSK} converges to $f^*$ and $\frac{1}{g}$. Let $C$ be a strictly positive symmetric matrix with unique matrices $P$ and $D$ satisfying the matrix balancing equation. Algorithm \ref{Alg::SSK} converges to $P$ and $D$.
\end{theorem}

Given a set of observations, a kernel density estimator constructed using kernel with infinite support is necessarily positive. This fact leads directly to the corollary below. 

\begin{corollary}\label{Th::KSKConv}
    Let $S=\left\{\left(X_k,Y_k)\right)\right\}_{k=1}^N$ be $N$ IID observations from a density on $[0,1]^2$ and $\widehat{f}$ be a kernel density estimator constructed using a kernel with infinite support. Let $\widehat{f}^*$ and $\widehat{g}$ be the unique functions that satisfy Equation \eqref{Eq::Cbal}. Algorithm \ref{Alg::KSK} converges to $\widehat{f}^*$ and $\frac{1}{\widehat{g}}$. 
\end{corollary}

These results highlight the first advantage of Algorithm $\ref{Alg::KSK}$ over Algorithm $\ref{Alg::SSK}$. The discrete algorithm requires strict positivity of the input to guarantee convergence, while the continuous algorithm will always converge, provided we use a kernel with infinite support. The Theorem \ref{Th::CSSKConv} follows a commonly used approach to study the matrix balancing problem \citep{BPS66}. The main strategy is then to view a matrix balancing algorithm as an operator acting on the cone of positive vectors. The theorem can then be proven by showing (1) the solution to the matrix balancing problem corresponds to a fixed point of this operator and (2) the operator is a contraction mapping.

Our main result characterizes the consistency of Algorithm \ref{Alg::KSK} and provides a bound on the error of the output in terms of the error of the input.

\begin{theorem}
\label{Th::FinResult}
	Assume the setting of Theorem~\ref{Th::CSSKConv}, and further assume that $f$ is twice differentiable with bounded, continuous, and square integrable partial derivatives. Let $S=\left\{\left(X_k,Y_k)\right)\right\}_{k=1}^n$ be $n$ IID samples from $f$. Suppose $\widehat{f}$ is a histogram estimator of $f$ based on $S$ using the theoretically optimal bin size, and let $\widehat{g}$ be obtained using Algorithm~$\ref{Alg::SSK}$. Then $\left\lVert\widehat{g}-g\right\rVert_2$ is $O_P(n^{-1/4})$. Similarly, suppose $\widehat{f}$ is a kernel density estimator of $f$ based on $S$ using a symmetric Gaussian kernel with the theoretically optimal bandwidth, and let $\widehat{g}$ be obtained using Algorithm~$\ref{Alg::KSK}$. Then $\left\lVert\widehat{g}-g\right\rVert_2$ is $O_P(n^{-1/3})$.
\end{theorem}

This result highlights the second advantage of Algorithm $\ref{Alg::KSK}$ over Algorithm $\ref{Alg::SSK}$ by characterizing the rate of convergence of $\widehat{g}$ to $g$. It is based on Theorem \ref{Th::KSKResult}, the main technical result of this paper. Intuitively, Theorem \ref{Th::KSKResult} states that for any estimator $\widehat{f}$ that converges to $f$, $g$ converges to $\widehat{g}$ at the same rate. The convergence rates of histogram and kernel density estimates are well-known, in terms of optimal mean integrated squared error (MISE) \citep{Simonoff96, ChaconDuong18}, leading to the theorem. 

\section{Numerical Results} 
\label{section::NR}
\subsection{Simulations}
We compare the performances of Algorithms \ref{Alg::SSK} (SSK) and \ref{Alg::KSK} (KSK) on a toy example. 
We design a symmetric probability density on the unit square with uniform marginals as follows. 
\begin{enumerate}
    \item Start with a mixture model, \[\tilde{f}\left(x,y\right)\propto\mathcal{N}\left(\mu,\Sigma\right)+\exp\left(-\frac{(x-y)^2}{0.01}\right)\] where $\mu=(0.2,0.8)^{\top}$ and $\Sigma=0.1I$, $I$ is $2\times 2$ identity matrix. 
    \item Symmetrize and normalize $\tilde{f}$ to obtain a SDSD  $f^*(x,y)$.
\end{enumerate}

We use a bias function  $g(x)=\cos(10\pi x)+3.5$ to create a distorted density $f(x,y)\propto g(x)f^*(x,y)g(y)$. The SDSD $f^*$ and the distorted  density are shown in Figures \ref{Fig::UnbiasedSamp} and \ref{Fig::BiasedSamp}, respectively. 

\begin{figure} \label{Fig::SampDist}
    \centering
    \begin{subfigure}[b]{0.35\textwidth}
        \centering
        \includegraphics[width=\textwidth]{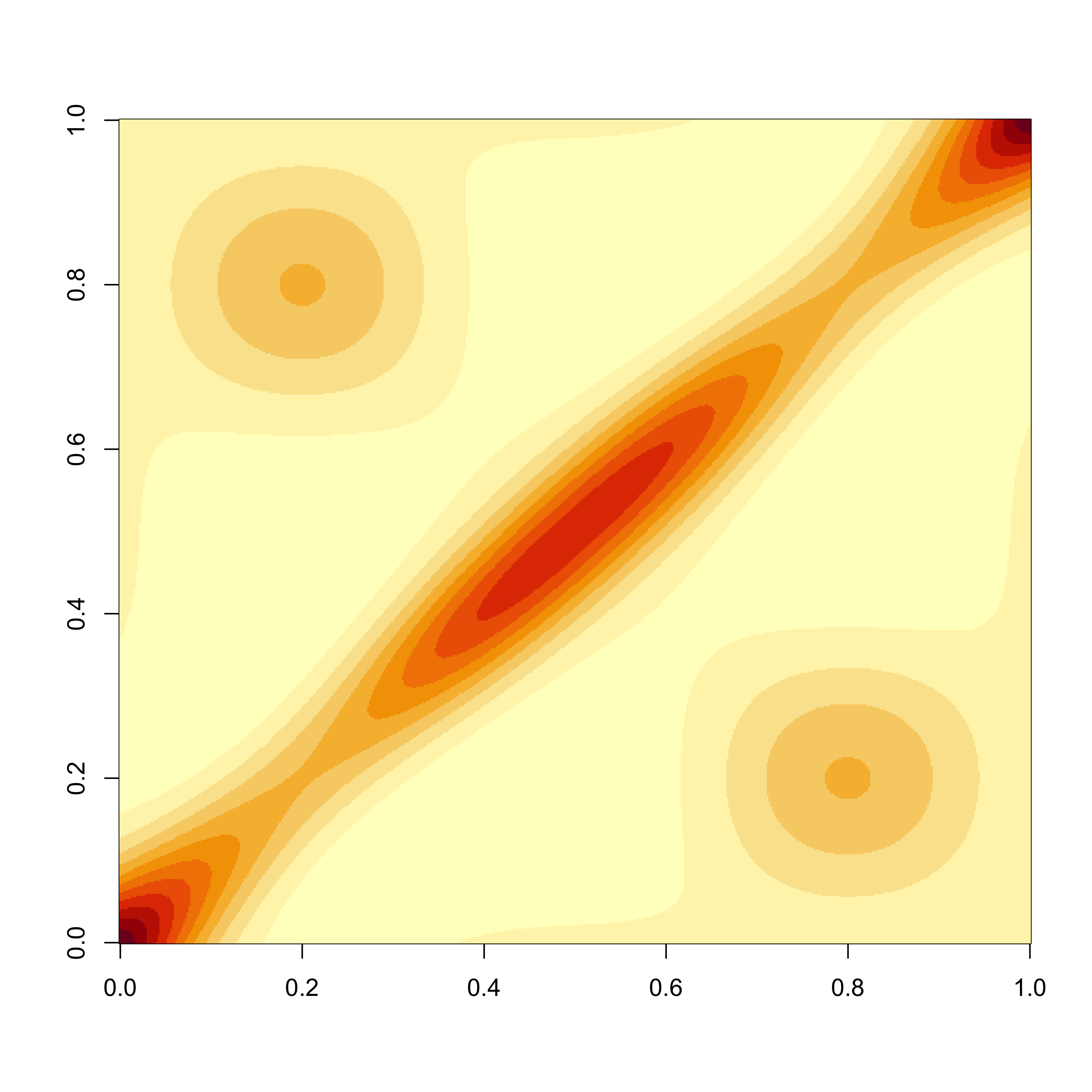}
        \caption{Unbiased density}
        \label{Fig::UnbiasedSamp}
    \end{subfigure}
    \begin{subfigure}[b]{0.35\textwidth}
        \centering
        \includegraphics[width=\textwidth]{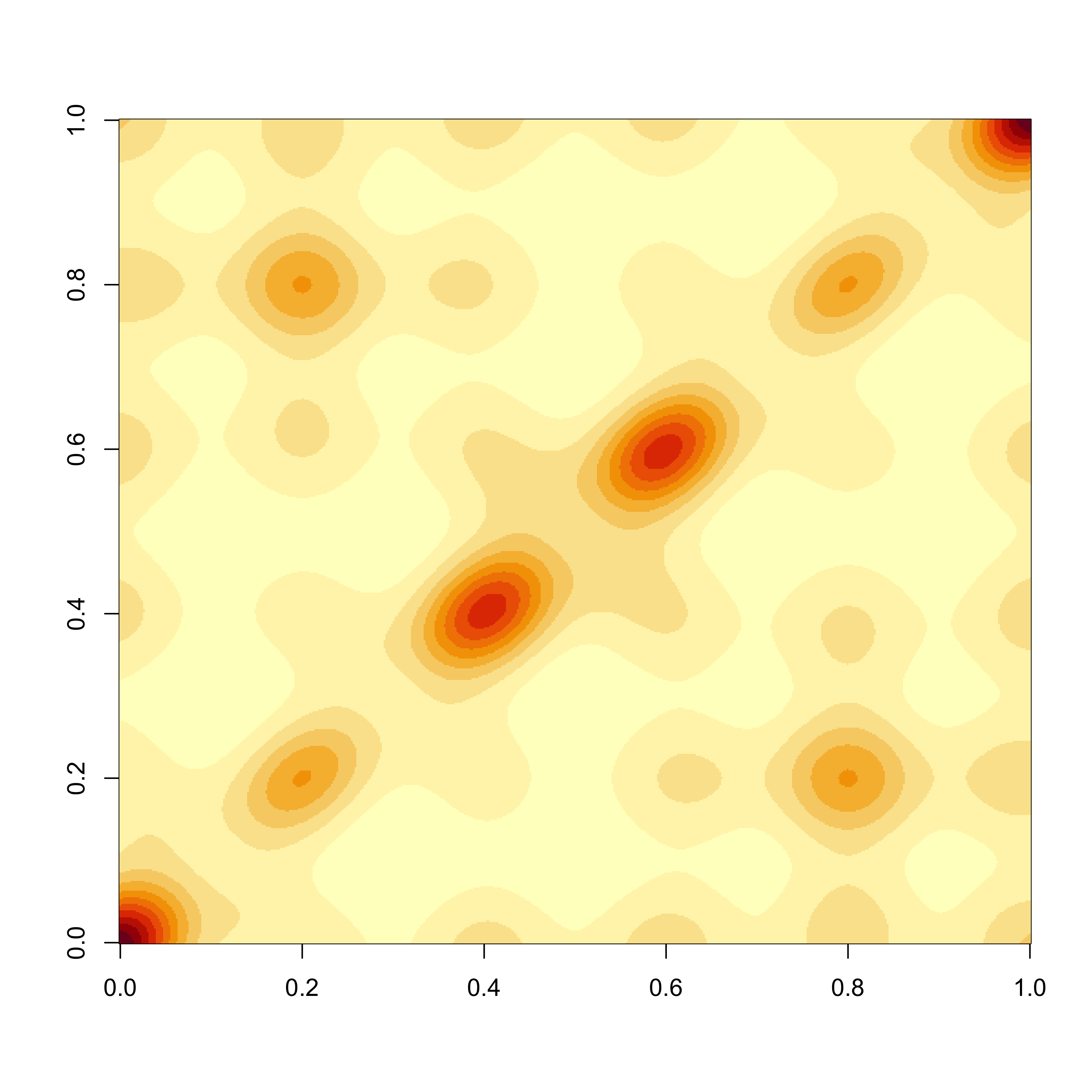}
        \caption{Biased density}
        \label{Fig::BiasedSamp}
    \end{subfigure}
    \caption{The image on the left shows the target density function, which has high values along the diagonal with a signal for interaction on the off-diagonal. The image on the right shows the function after biasing. In the biased data, several false peaks are visible which obscure the original signal.} 
\end{figure}

Figure \ref{Fig::Compare1} compares the performance of the SSK and KSK algorithms on one simulated dataset with about $65000$ observations. \ref{Fig::gEstimates} displays the true bias function as a solid black line. The KSK estimate is shown in blue and the SK estimate in red. In this example, the bandwidth selected was $h\approx0.03$ and the bin width selected was $1/109\approx0.009.$ In general, the bandwidths selected by KSK were larger than the corresponding bin widths selected by SK, highlighting the benefits of utilizing the spatial information inherent in the problem. 

Figures \ref{Fig::SKFin} and \ref{Fig::KSKFin} display the de-biased estimates of the balanced density $f^*$. The KSK estimate shows clearer evidence of a single peak on the off-diagonals location near the original signal.  In either case, the simulation demonstrates that under our statistical framework, the process estimating the bias function using a Sinkhorn-like algorithm, then using this estimate to recover the balanced density can work well. 

\begin{figure}
    \centering
    \begin{subfigure}{0.32\textwidth}
        \includegraphics[width=\linewidth]{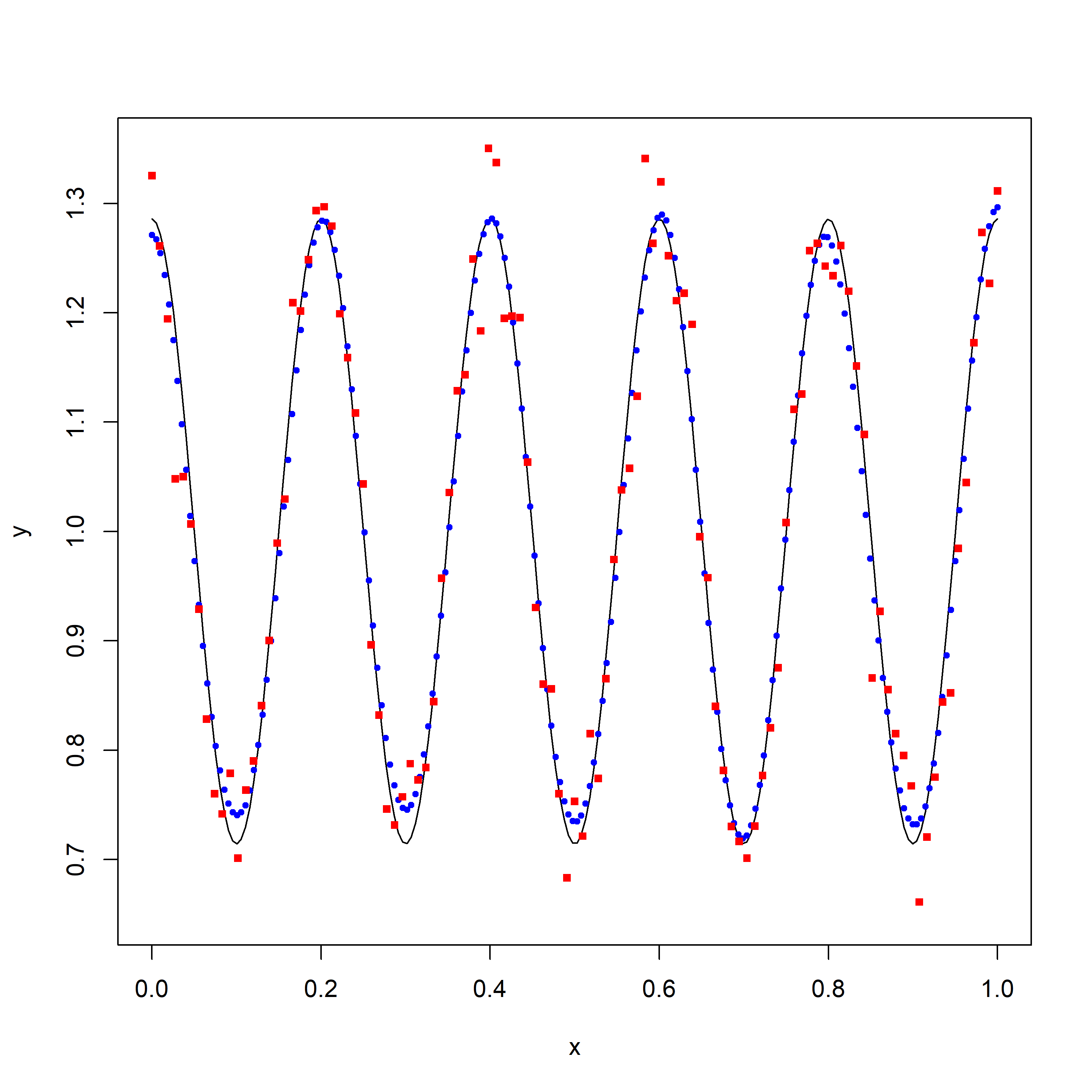}
        \caption{Estimates for Bias}
        \label{Fig::gEstimates}
    \end{subfigure}
    \begin{subfigure}[b]{0.32\textwidth}
        \centering
        \includegraphics[width=\textwidth]{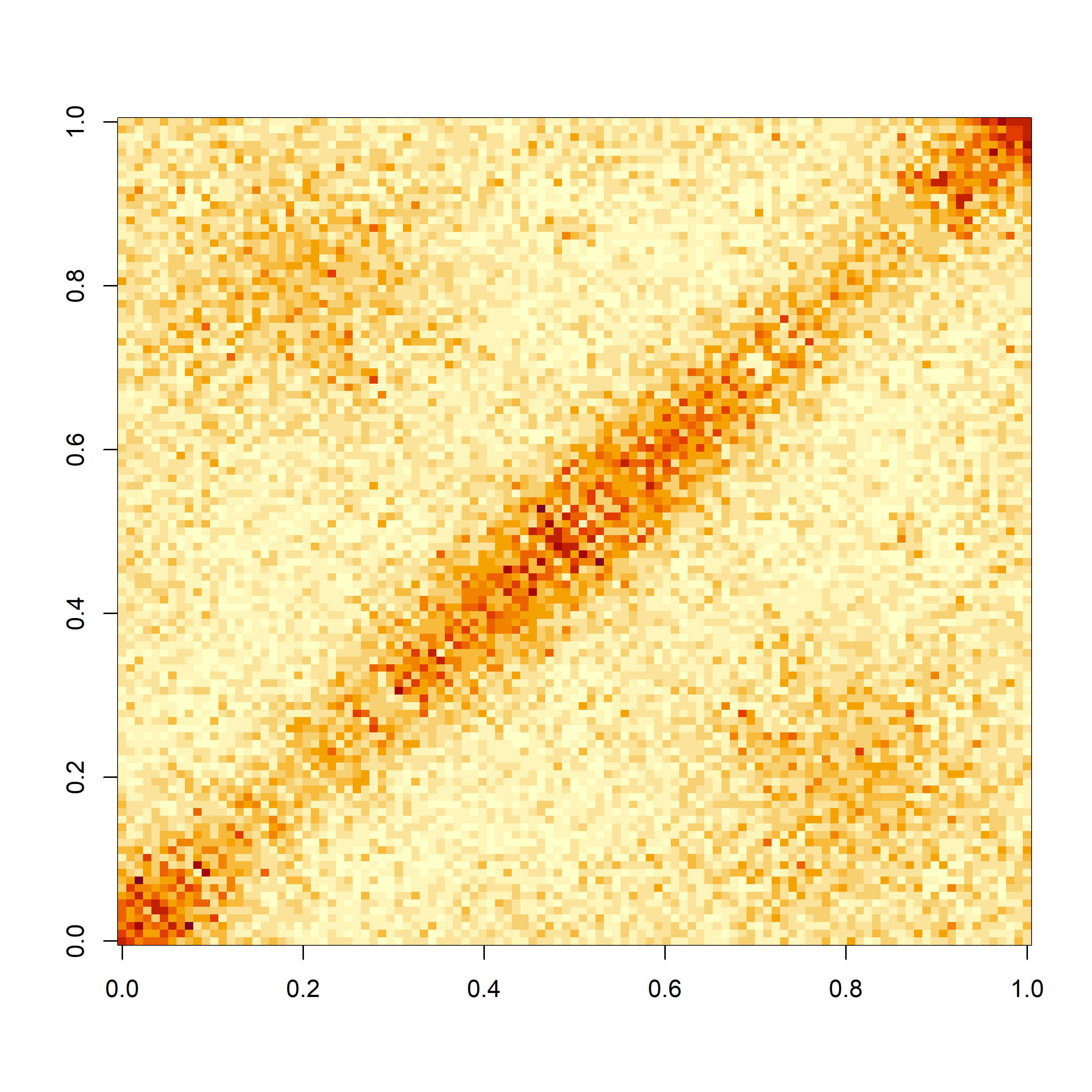}
        \caption{SK De-biased Estimate}
        \label{Fig::SKFin}
    \end{subfigure}
    \begin{subfigure}[b]{0.32\textwidth}
        \centering
        \includegraphics[width=\textwidth]{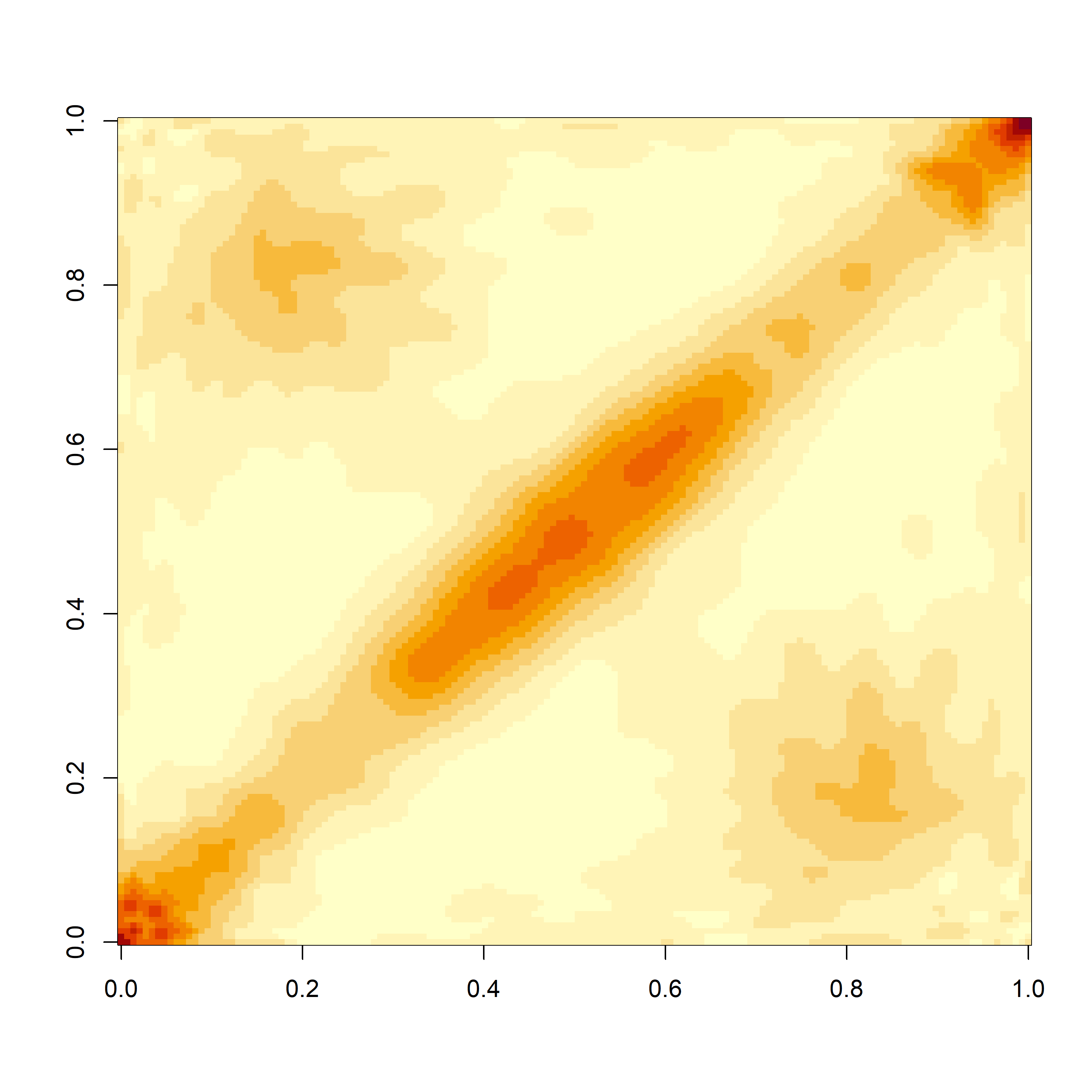}
        \caption{KSK De-Biased Estimate}
        \label{Fig::KSKFin}
    \end{subfigure}
    \caption{A comparison of the SK and KSK algorithms for one sample of about 60000 observations. (a) The true bias function, shown in black and the results of KSK (blue) and SSK (red). (b) and (c) The recovered signals obtained by applying the estimated bias functions to the original data, using the optimal bin size or bandwidth. The KSK estimate shows clearer visual evidence a peak near the original signal location, $(0.2,0.8).$}
	\label{Fig::Compare1}
\end{figure}

Figure \ref{Fig::Compare2} compares the overall performance the SK and KSK algorithms on many runs across a range of scales. In each subfigure, the KSK results are shown in blue and the SSK results are shown in red. Figure \ref{Fig::MSE} displays the average MISE over $30$ runs at fixed numbers of points. This was repeated for $13$ different sample sizes. Both algorithms produce better estimates of the bias function as the sample size increases, but the error from the KSK estimate decreases more quickly and is generally much smaller than the error from the SK estimate. Figure \ref{Fig::KSKFin} displays the average bandwidth size from the same runs ($1$ divided by the number of bins for the SK algorithm). Again as the sample size increases, the bandwidth decreases, but the SK ``bandwidth" is smaller than the KSK bandwidth, suggesting there is spatial data leveraged by the KSK algorithm that the SK algorithm excludes.  

Figure \ref{Fig::MSEvsCV} shows a comparison of the minimal MSE achieved and MSE obtained using cross-validation. The data shown consists of about $200$ runs for a range of $20000-75000$ observations. The line displayed is the line $x=y$, which is the lower boundary of any point displayed. The figure shows that the KSK algorithm consistently chooses a bandwidth that is closer to the ``optimal'' bandwidth when compared to the SSK algorithm. It also again shows that the KSK algorithm tends to have a lower MSE in general. The SSK algorithm also has higher variance, sometimes choosing bandwidths that are quite far from the optimal choice. 

\begin{figure}[htb]
    \centering
    \begin{subfigure}{0.32\textwidth}
        \includegraphics[width=\linewidth]{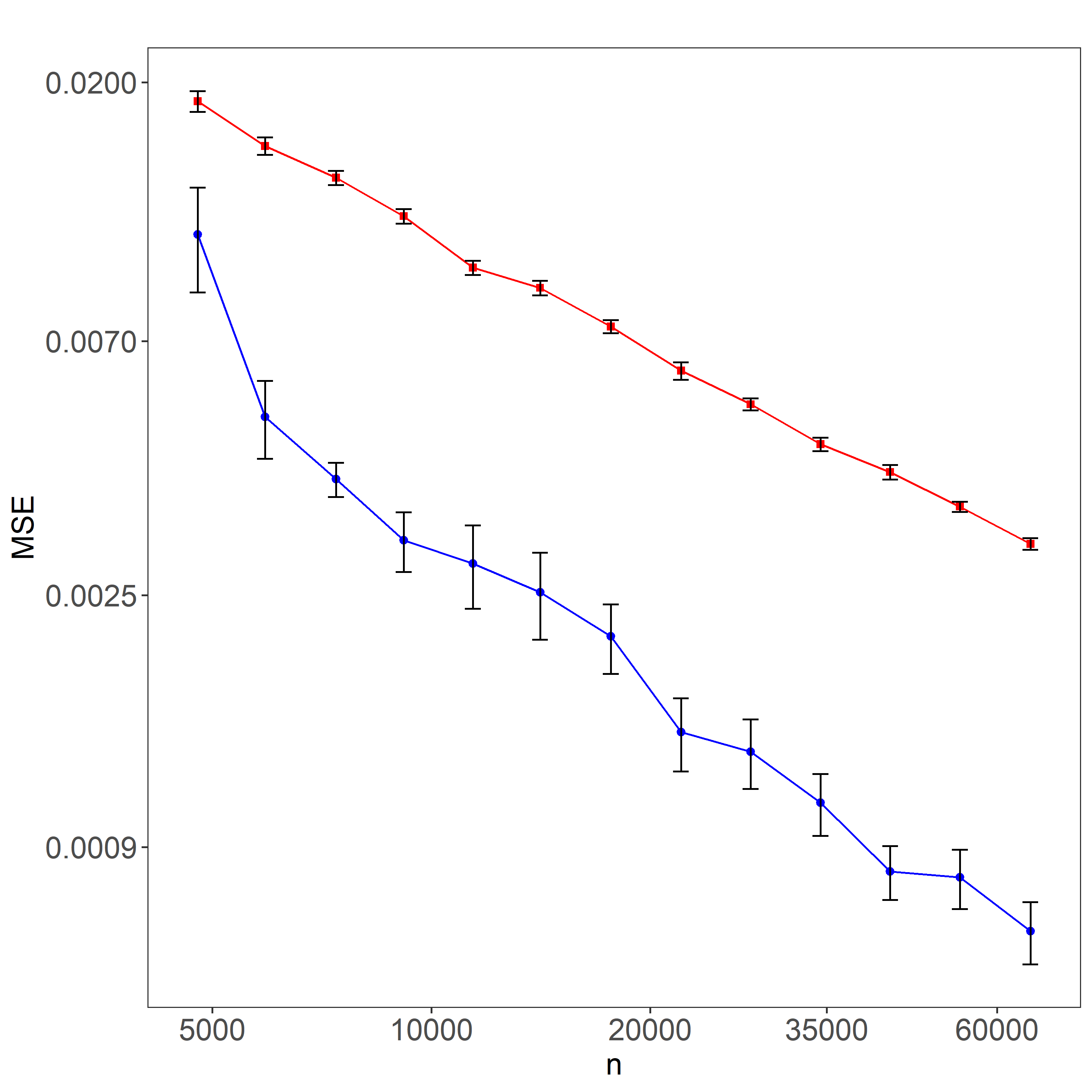}
        \caption{}
        \label{Fig::MSE}
    \end{subfigure}
        \begin{subfigure}{0.32\textwidth}
        \includegraphics[width=\linewidth]{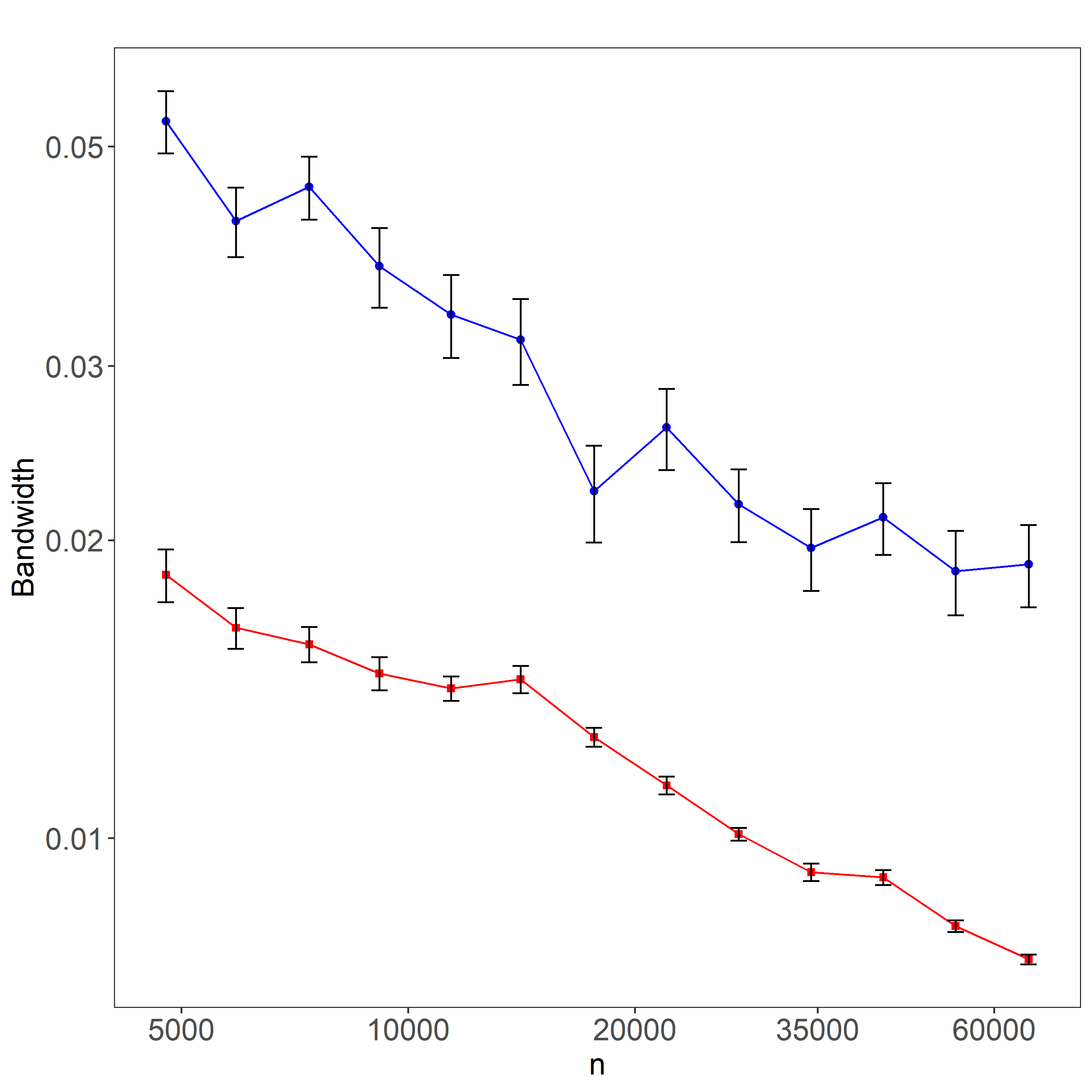}
        \caption{}
        \label{Fig::Bandwidth}
    \end{subfigure}
    \begin{subfigure}{0.32\textwidth}
        \includegraphics[width=\linewidth]{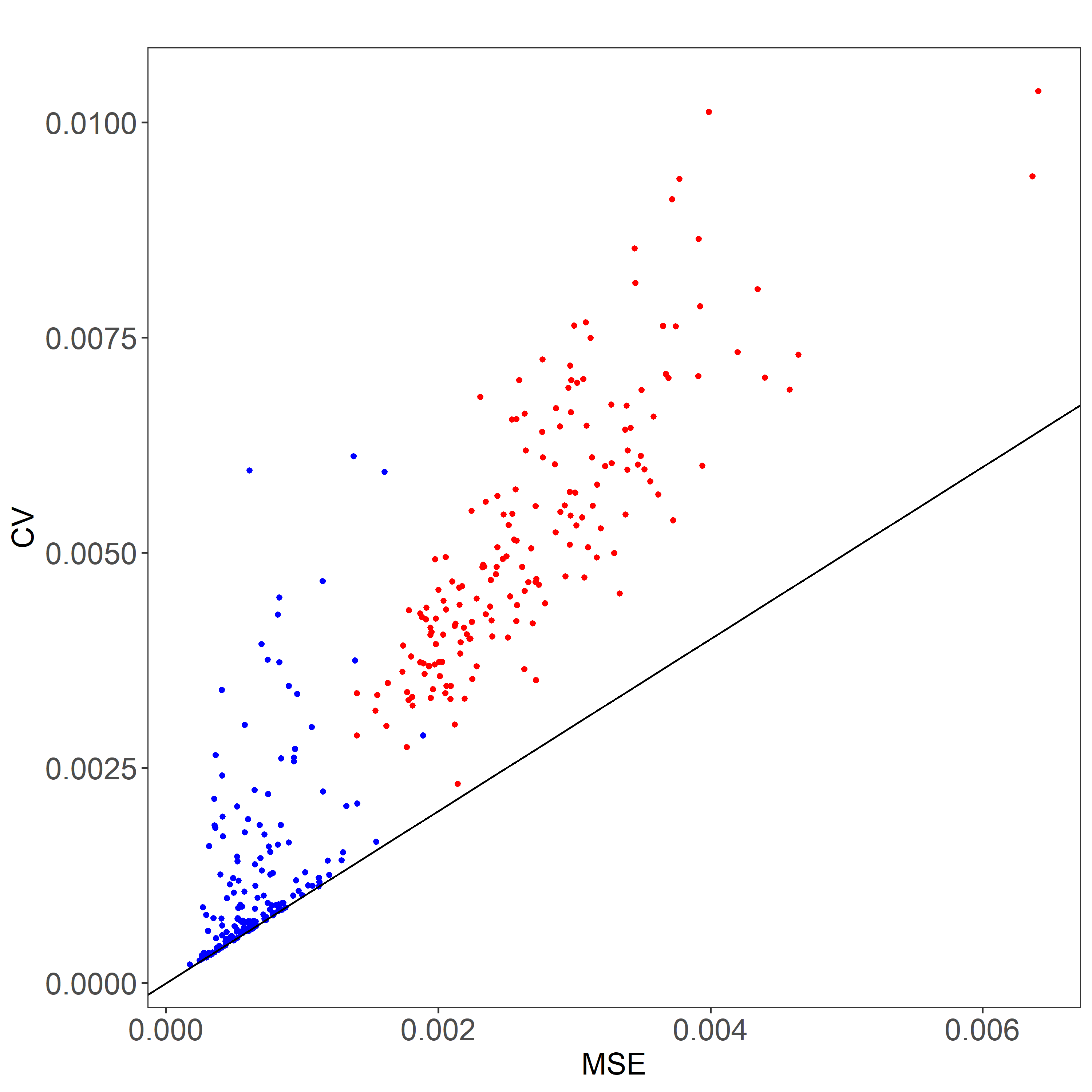}
        \caption{}
        \label{Fig::MSEvsCV}
    \end{subfigure}
    \caption{Results from simulated data for varying $n$. In each case, the results of KSK are shown in blue and results of SSK are shown in red. (a) The MISE of the estimator relative to the true function. (b) The bandwidth selected by the algorithm. For SSK, the inverse of the bin size is used. (c) Comparison of the MISE obtained using CV-selected bandwidths and the optimal MISE achieved using the bandwidth minimizing error relative to the true function.}
    \label{Fig::Compare2}
\end{figure}

\subsection{IMR90 Data}
Algorithms \ref{Alg::KSKCV} and \ref{Alg::SSKCV} were run using Hi-C data at the resolution level $25$kb \citep{Rao14}. At this resolution, the dimension of the Hi-C matrix is $10000$. In the discrete case, the input data was binned by scaling the matrix by factors in $\left\{1,2,3,5,8,11,13\right\}$ and in the continuous case, $35$ bandwidths in $\left\{0.003,0.0035,\ldots,0.02\right\}$ were tested to determine the optimal parameters.

\begin{figure}
\centering
	\begin{tabular}{|c||c|c|c|c|c|c|c|c|c|c|}
		\hline
		Chromosome&1&2&3&4&5&6&7&8&9&10\\
		\hline
		Discrete&2&2&1&1&2&2&1&3&2&1\\
		\hline
		Continuous&0.02&0.017&0.004&0.008&0.005&0.08&0.01&0.04&0.005&0.005\\
		\hline
	\end{tabular}
    \caption{The optimal results from cross-validation on Hi-C data}
\end{figure}

One of the benefits of using the KSK algorithm is that, if the kernel is chosen to have infinite support, the resulting KDE is strictly positive and existence of a unique solution to the balancing problem is guaranteed. Thus, the KSK algorithm will always converge to a solution provided that such a kernel is used. This is not true in the case of Hi-C data, which is typically extremely sparse and becomes worse as the resolution of the data increases. This can cause the KR algorithm to fail, as can be seen in Chromosome $9$ at this resolution. In these cases, the matrix may be scaled to decrease the sparsity factor. Even when this occurs, the KSK algorithm again chooses bandwidths much larger than the corresponding bin width selected by the SK algorithm, suggesting that the spatial information present in the data plays a role in the quality of the final estimate. 

\bibliographystyle{apalike}
\bibliography{Report.bib} 

\section{Appendix}
\subsection{Lemmas and Proofs}
\subsubsection{Proof of Lemma \ref{lemma1}}
The kernel density estimator $\hat{f}$ is given by 
\[\hat{f}=\frac{1}{Nh^2}\sum_{k=1}^NK_2\left(\frac{x-X_k}{h},\frac{y-Y_k}{h}\right)=\frac{1}{Nh^2}\sum_{k=1}^N\left[K\left(\frac{x-X_k}{h}\right)K\left(\frac{y-Y_k}{h}\right)\right].\]
Then 
\begin{align*}
\int_{-\infty}^\infty\hat{f}(x,y)dy&=\int_0^1\frac{1}{Nh^2}\sum_{k=1}^NK\left[\left(\frac{x-X_k}{h}\right)K\left(\frac{y-Y_k}{h}\right)\right]dy\\
&=\frac{1}{Nh^2}\sum_{k=1}^N\left[K\left(\frac{x-X_k}{h}\right)\int_{-\infty}^\infty K\left(\frac{y-Y_k}{h}\right)dy\right]\\
&=\frac{1}{Nh^2}\sum_{k=1}^N\left[K\left(\frac{x-X_k}{h}\right)h\right]\\
&=\frac{1}{Nh}\sum_{k=1}^NK\left(\frac{x-X_k}{h}\right)
\end{align*}

If $f$ is symmetric, we may symmetrize the data to obtain $2N$ observations. The corresponding symmetric estimator is 
\[\hat{f}=\frac{1}{2Nh^2}\sum_{k=1}^N\left[K_2\left(\frac{x-X_k}{h},\frac{y-Y_k}{h}\right)+\sum_{k=1}^NK_2\left(\frac{x-Y_k}{h},\frac{y-X_k}{h}\right)\right].\]
The final statement then follows by applying the previous argument to this estimator. 

\subsubsection{Proof of Theorem \ref{Th::CSSKConv}}
\begin{proof}
We begin with the second claim. Two main ideas from non-linear Perron-Frobenius theory  are explicitly used. See \citep{Nussbaum12} for a rigorous treatment of the subject. First, Hilbert's projective metric on $\mathbb{R}^n_+$ is defined as
\begin{equation*}
	d(x,y)=\log{\frac{\max_i\left\{ x_i/y_i\right\}}{\min_i\left\{x_i/y_i\right\}}}.
\end{equation*}
This may be equivalently written
\begin{equation*}
	d(x,y)=\left\lVert \log(x)-\log(y)\right\rVert_V
\end{equation*}
where 
\begin{equation*}
	\left\lVert x\right\rVert_V=\max_i x_i-\min_i x_i.
\end{equation*} 
Second, a theorem of Birkhoff's that states that with respect to this metric, any positive matrix $C$ is a contraction \citep{Birkhoff57}. 

As stated in Section \ref{Section::TR}, one way to study matrix balancing algorithms is to view them as operators acting on the cone of positive vectors. Consider Algorithm~\ref{Alg::SSK}. Letting $x^{(0)}=e$, elementary vector manipulation shows that the next iterate can be obtained using the formula 
\begin{equation}
\label{Eq::SSKIt}
x^{(t+1)}=\sqrt{\frac{x^{(t)}}{Cx^{(t)}}}.
\end{equation}

We then define an operator \[T_Cx:=\sqrt{\frac{x}{Cx}}\] and relate the fixed points of $T_C$ to the matrix balancing problem. If $D$ is the balancing matrix for $C$ and $d$ is the balancing vector, then $DCDe=DCd=e$. We may rearrange this equation to get 
\begin{align*}
	DCDe&=e\\
	DCd&=e\\
	Cd&=\left(D\right)^{-1}e\\
	\left(Cd\right)^{-1}&=d.
\end{align*}
Then
\begin{equation*}
	d=\sqrt{\frac{d}{Cd}}=Td
\end{equation*}
and $d$ is the balancing vector of $C$ if and only if $d$ is a fixed point of $T_C$.

Suppose that $T_C$ is a contraction under Hilbert's projective metric. This would imply $T_C$ converges in direction. Since $T_C$ is constant on rays from the origin, this implies that $T_C$ converges to a fixed point.  Thus, it suffices to show that $T_C$ is a contraction under Hilbert's projective metric. Define three operators
\begin{equation*}
	Sx=\left(Cx\right)^{-1}\qquad Rx=x* Sx\qquad Qx=\sqrt{x}
\end{equation*}
so that $T_Cx=Q\circ R\left(x\right)$.

First, note that 
\begin{align*}
	d\left(1/x,1/y\right)=d\left(x,y\right)
\end{align*}
so 
\begin{align*}
	d\left(Sx,Sy\right)=d\left(Cx,Cy\right)<d\left(x,y\right)
\end{align*}
where the last inequality follows from Birkhoff's theorem. 
Then
\begin{align*}
	d\left(Rx,Ry\right)&=d\left(x\odot Sx,y\odot Sy\right)\\
					     &=\left\lVert \log(x\odot Sx)-\log(y\odot Sy)\right\rVert_V\\
					     &=\left\lVert \log(x)-\log(y)\right\rVert_V+\left\lVert \log(Sx)-\log(Sy)\right\rVert_V\\
					     &=d\left(x,y\right)+d\left(Sx,Sy\right)<2d\left(x,y\right).
\end{align*}
Finally, 
\begin{align*}
	d\left(Qx,Qy\right)=\frac{1}{2}d\left(x,y\right).
\end{align*}
Thus, 
\begin{align*}
	d\left(T_Cx,T_Cy\right)<\frac{1}{2}\left(2d\left(x,y\right)\right)<d\left(x,y\right)
\end{align*}
and $T_C$ is a contraction. 

The continuous case then follows using similar ideas. In Algorithm~\ref{Alg::CSSK}, letting $g^{(0)}(x)=1$, the next iterate is \[g^{(t+1)}(x)=\sqrt{\frac{g^{(t)}(x)}{\int_0^1f(x,y)g^{(t)}(x)dx}}.\]Let
\begin{equation}
      Sg=\left(\int_0^1f(x,y)g(x)dx\right)^{-1}\qquad Rg=g\cdot Sg\qquad Qg=\sqrt{g}
\end{equation}
so that $Tg=Q\circ R\left(g\right)$. It then suffices to show that $S$ remains a contraction in the continuous case, which is shown in \citep{Nussbaum93}.
\end{proof}

\subsubsection{Theorem \ref{Th::KSKResult} and Proof} 
\begin{theorem}
\label{Th::KSKResult}
	Let $f(x,y)$ and $\widehat{f}(x,y)$ be bounded, continuous, strictly positive, symmetric densities on $[0,1]^2$ such that $\left\lVert f(x,y)-\widehat{f}(x,y)\right\rVert_2<\varepsilon$. Let $g$ and $\widehat{g}$ be the balancing functions for $f(x,y)$ and $\widehat{f}(x,y)$, respectively, i.e., 
    \begin{align*}
	\int_0^1g(x)f(x,y)g(y)dx=\int_0^1\widehat{g}(x)\widehat{f}(x,y)\widehat{g}(y)dx&=1,\\
        \int_0^1g(x)f(x,y)g(y)dy=\int_0^1\widehat{g}(x)\widehat{f}(x,y)\widehat{g}(y)dy&=1
	\end{align*} 
    and let $M=\max{\left(g,\widehat{g}\right)}$. Then $\left\lVert g-\widehat{g}\right\rVert_2<M^2\varepsilon$. 
\end{theorem}

\begin{proof}
	First, note that $g$ and $\widehat{g}$ must be bounded. Let $M=\max{\left\{g,\widehat{g}\right\}}$. Then 
	\begin{align*}
		\left\lVert 1-\int_0^1g(x)\widehat{f}(x,y)g(y)dx\right\rVert_2&=\left\lVert \int_0^1g(x)f(x,y)g(y)dx-\int_0^1g(x)\widehat{f}(x,y)g(y)dx\right\rVert_2\\
										&=\left\lVert \int_0^1g(x)\left(f(x,y)-\widehat{f}(x,y)\right)g(y)dx\right\rVert_2<M^2\varepsilon. 
	\end{align*}
	Define an operator $T_{\widehat{f}}:\mathcal{B}_+\left(\left[0,1\right]\right)\to \mathcal{B}_+\left(\left[0,1\right]\right)$ that maps the cone of bounded positive functions on $[0,1]$ to itself, \[
		T_{\widehat{f}}(g)=\int_0^1g(x)\widehat{f}(x,y)g(y)dx. 
	\]
	The above equation can then be rewritten \[
		\left\lVert T_{\widehat{f}}\left(\widehat{g}\right)-T_{\widehat{f}}\left(g\right)\right\rVert_2<M^2\varepsilon. 
	\]

	Given a fixed function $\widehat{f}$, $T_{\widehat{f}}$ returns the marginals of the scaled function $g(x)\widehat{f}(x,y)g(y)$. Theorem \ref{Th::KSK} then implies that $T_{\widehat{f}}$ is invertible.
    
	Showing that $T_{\widehat{f}}^{-1}$ is Lipschitz continuous at $\widehat{g}$ would then complete the proof. Thus, it suffices to show that the operator norm of the Fr\'{e}chet derivative is nonzero. The Fr\'{e}chet derivative may be computed explicitly as 
	\begin{align*}
		DQ_{\widehat{g}}(h)=\widehat{g}(y)\int_0^1h(x)\widehat{f}(x,y)dx+h(y)/\widehat{g}(y). 
	\end{align*}
	For a fixed $\widehat{g}$, the operator norm of $DQ_{\widehat{g}}(h)$ is clearly nonzero. 
\end{proof}

\subsubsection{Calculation of the Fr\'{e}chet Derivative}
The Fr\'{e}chet derivative of an operator $T$ is defined to be a bounded linear operator $A$ such that \[
	\lim_{\left\lVert h\right\rVert_2\to0}\frac{\left\lVert T(g+h)-T(g)-Ah\right\rVert_2}{\left\lVert h\right\rVert_2}=0. 
\]
If \[
	T(g)=\int_0^1g(x)f(x,y)g(y)dx
\]
then letting \[
	A(h)=g\left(y\right)\int_0^1h\left(x\right)f\left(x,y\right)dx+h\left(y\right)\int_0^1g\left(x\right)f\left(x,y\right)dx, 
\]
\begin{align*}
	\lim_{\left\lVert h\right\rVert_2\to0}&\frac{\left\lVert T(g+h)-T(g)-Ah\right\rVert_2}{\left\lVert h\right\rVert_2}\\
	&=\lim_{\left\lVert h\right\rVert_2\to0}\frac{1}{\left\lVert h\right\rVert_2}\left\lVert\int_0^1\left(g(x)+h(x)\right)f\left(x,y\right)\left(g(y)+h(y)\right)dx+\cdots\right.\\
	&\cdots\left.-\int_0^1g(x)f(x,y)g(y)dx-g(y)\int_0^1h(x)f(x,y)dx-h(y)\int_0^1g(x)f(x,y)dx\right\rVert_2\\
	&=\lim_{\left\lVert h\right\rVert_2\to0}\frac{\left\lVert\int_0^1h(x)f(x,y)h(y)dx\right\rVert_2}{\left\lVert h\right\rVert_2}\\
	&\leq\lim_{\left\lVert h\right\rVert_2\to0}\frac{\left\lVert h\right\rVert_2^2\left\lVert f(x,y)\right\rVert_2}{\left\lVert h\right\rVert_2}=0. 
\end{align*}

\subsection{Classical Results for Matrix Balancing} \label{Section::KnownThms}
In this subsection, we collect some existing results on matrix balancing. Theorems \ref{Th::SK} and \ref{Th::ASK} relate to the matrix balancing and matrix scaling problems, respectively. They establish sufficient conditions for the existence and uniqueness of solutions, and provide simple iterative algorithms that are guaranteed to converge \citep{Sinkhorn64, Sinkhorn67}. Theorem \ref{Th::KSK}, addressing the continuous density scaling problem, offers analogous results in a more general setting \citep{BLN94}.

\begin{theorem} \label{Th::SK}
Given a strictly positive $n\times n$ matrix $C$ there is a unique doubly stochastic matrix $P=D_1CD_2$ where $D_1$ and $D_2$ are diagonal matrices with positive diagonals. The matrices $D_1$ and $D_2$ are unique up to a scalar factor. That is, if $D'_1$ and $D'_2$ are positive diagonal matrices such that $D'_1CD'_2$ is also doubly stochastic, then $D_1=\alpha D'_1$ and $\alpha D_2=D'_2$ for $\alpha>0$. Additionally, Algorithm \ref{Alg::SK} converges to $P$. 
\end{theorem}

\begin{theorem} \label{Th::ASK}
Given a strictly positive $n\times n$ matrix $C$ and two positive $n\times 1$ column vectors $r$ and $c$ such that $r^{\top}e=c^{\top}e$, there exist diagonal matrices $D_1$ and $D_2$ such that for $P:=D_1CD_2$,
    \begin{align*}
    Pe&=r\\
    P^{\top}e&=c.
    \end{align*} 
    The matrices $D_1$ and $D_2$ are unique up to a scalar factor. That is, if $\tilde{D}_1$ and $\tilde{D}_2$ are diagonal matrices such that $\tilde{D}_1C\tilde{D}_2$ also has the specified row and column sums, then $D_1=\alpha \tilde{D}_1$ and $\alpha D_2=\tilde{D}_2$. Additionally, the process of alternatively scaling the rows and columns of $C$ to $r$ and $c$, respectively, is an algorithm that converges to $P$. 
\end{theorem}

\begin{theorem}
    \label{Th::KSK}
    Given a strictly positive density $f(x,y)$ on the unit square and two positive functions $r$ and $c$, there exist positive functions $h_1(x)$ and $h_2(y)$ such that for $p(x,y):=h_1(x)f(x,y)h_2(y)$,
    \begin{align*}
    \int_0^1p(x,y)dx&=r(y)\\
    \int_0^1p(x,y)dy&=c(x).
    \end{align*}
    The functions $f$ and $g$ are unique up to a scalar factor. That is, if $h_1'$ and $h_2'$ are diagonal matrices such that $h'_1(x)f(x,y)h_2'(y)$ also has the specified marginals, then $h_1'=\alpha h_1$ and $\alpha h_2'=h_2$. Additionally, the process of alternatively scaling $f$ so its marginals are $r(x)$ and $c(y)$ is an algorithm that converges to $p$. 
    
    In particular, if $f$ is symmetric, i.e. $f(x,y)=f(y,x)$ and $m=r(x)=c(y)$ there exists a unique positive function $h$ such that for $p(x,y):=h(x)f(x,y)h(y),$
    \begin{align*}
    \int_0^1p(x,y)dx&=m(y)\\
    \int_0^1p(x,y)dy&=m(x).
    \end{align*}
\end{theorem}

\subsection{A Result on the Consistency of Algorithm \ref{Alg::SSK}}
This theorem is a discrete analogue of Theorem \ref{Th::KSKResult}. This algorithm is not directly applicable to the Hi-C situation, as the underlying density is continuous. However, it may be of interest when the quantity to be estimated is also a matrix. The balancing equation are scaled by $n$ to ensure that the entries of the marginals sum to $1$.

\begin{theorem}\label{Th::SKResult}
	Let $C$ and $\widehat{C}$ be strictly positive, symmetric $n\times n$ matrices such that $\min c_{ij}=\min \hat{c}_{ij}=1$, $\max{c_{ij},\hat{c}_{ij}}<M$, and $\left\lVert C-\widehat{C}\right\rVert_F<\frac{\varepsilon}{n}$. Let $D$ and $\widehat{D}$ be the balancing matrices for $C$ and $\widehat{C}$, respectively, $DCDe=\widehat{D}\widehat{C}\widehat{D}e=\frac{e}{n}$. Then $\left\lVert D-\widehat{D}\right\rVert_F<M\varepsilon$. 
\end{theorem} 

\subsubsection{Proof of Theorem \ref{Th::SKResult}}
We require two supporting lemmas, which are proven below.

\begin{lemma}\label{lemma::AbsBnd}
Suppose $C$ is an $n\times n$ positive, symmetric matrix such that $\min c_{ij}=1$ and $\max c_{ij}<M$ for some absolute constant $M$. Let $D$ be the balancing matrix for $C$. Then $\sqrt{n/M}<d_i<\sqrt{n}M$. 
\end{lemma}

\begin{lemma}\label{lemma::GCT}
	Let $D$ and $C$ be square matrices such that $D$ is diagonal and $P=DCD$ is a positive symmetric doubly stochastic matrix. The smallest eigenvalue of $P$ is bounded below by \[\min_i\left(2d^2_{ii}c_{ii}\right).\]
\end{lemma}

\begin{proof}
    First, note that 
	\begin{equation}\label{opcond}
		\left\lVert \frac{e}{n}-D\widehat{C}De\right\rVert_F=\left\lVert DCDe-D\widehat{C}De\right\rVert_F=\left\lVert D\left(C-\widehat{C}\right)De\right\rVert_F<M\varepsilon,
	\end{equation}
	where the last inequality follows from Lemma \ref{lemma::AbsBnd}. Define an operator $T_{\widehat{C}}:\mathbb{R}^n\to\mathbb{R}^n$ \[
		T_{\widehat{C}}(x)=X\widehat{C}Xe.
	\]
	Given a fixed matrix $\widehat{C}$, $T$ returns the marginals of the scaled matrix $X\widehat{C}X$. Theorems \ref{Th::SK} and \ref{Th::ASK} then imply that $T$ is invertible. 

	Since, $\widehat{D}\widehat{C}\widehat{D}e=e/n,$ Equation \ref{opcond} can be rewritten \[
		\left\lVert T_{\widehat{C}}(\widehat{d})-T_{\widehat{C}}\left(d\right)\right\rVert_F<M\varepsilon. 
	\]

	Showing that $T_{\widehat{C}}^{-1}$ is Lipschitz continuous at $x=\widehat{d}$ would then complete the proof. By the inverse function theorem, it suffices to show that the Jacobian of $T_{\widehat{C}}$ is bounded below by some absolute positive constant this point. The Jacobian matrix can be explicitly calculated to be \[
		J_{\widehat{C}}(x)=\widehat{X}\widehat{C}+\text{diag}\left(\widehat{C}x\right). 
	\]
	Since $\text{diag}\left(\widehat{C}\widehat{d}\right)=1/\widehat{D},$  \[
        J_{\widehat{C}}\left(\widehat{d}\right)=\left(\widehat{D}\widehat{C}\widehat{D}+I\right)\left(\widehat{D}\right)^{-1}.
	\]
	
	Lemma \ref{lemma::GCT} bounds the eigenvalues of $\widehat{D}\widehat{C}\widehat{D}+I$ by \[\min_i(2\hat{d}^2_{ii}\widehat{c}_{ii}-1)>0.\]

	Since $\widehat{D}$ is a diagonal matrix, the eigenvalues of $J_{\widehat{C}}\left(\widehat{d}\right)$ are bounded below by \[\min_j\left(\frac{1}{\widehat{d}_j}\right)\cdot\min_i(2\hat{d}^2_i\widehat{c}_{ii}-1).\]

  	It now suffices to bound $\widehat{d}_i$ and $\widehat{c}_{ii}$. The bound on the former follows from Lemma \ref{lemma::AbsBnd} and the bound on the latter follows from the assumed conditions.
\end{proof}

\subsubsection{Proof of Lemma \ref{lemma::AbsBnd}}
\begin{proof}
The balancing equation is \[\sum_{j=1}^nc_{ij}d_id_j=1.\]  This implies that 
\begin{equation}\label{Eq::Maxeq}
d_i\sum_jd_j\leq 1\leq Md_i\sum_jd_j.
\end{equation}
so 
\begin{equation}\label{Eq::Deq}
	\frac{1}{M\sum_j d_j}\leq d_i\leq \frac{1}{\sum_j d_j}
	\end{equation}
Summing inequality \ref{Eq::Maxeq} with respect to $i$ gives 
\begin{align*}
\left(\sum_id_i\right)^2&\leq n\leq M\left(\sum_id_i\right)^2\\
\sum_id_i&\leq \sqrt{n}\leq\sqrt{M}\sum_id_i
\end{align*}
Combining this with inequality \ref{Eq::Deq} completes the proof. 
\end{proof}

\subsubsection{Proof of Lemma \ref{lemma::GCT}}
\begin{proof}
	Let $A$ be a matrix. An interval centered at $a_{ii}$ with radius $\sum_{j\neq i}\left\lvert a_{ij}\right\rvert$ is called a Gershgorin disk. The Gershgorin circle theorem states that any eigenvalue must lie in one of these disks. When applied so positive symmetric doubly stochastic matrices, this implies that the smallest eigenvalue is bounded below by 
    \begin{equation*}
        \min_i\left(p_{ii}-\sum_{j\neq i}p_{ij}\right)=\min_i\left(p_{ii}-\left(1-p_{ii}\right)\right)=\min_i\left(2p_{ii}-1\right).
    \end{equation*}

	 \[\min_i(2\widehat{d}^2_i\widehat{c}_{ii}-1)>0.\] Finally, since $\widehat{D}$ is a diagonal matrix, the eigenvalues of $J_{\widehat{C}}\left(\widehat{d}\right)$ are bounded below by \[\min_j\left(\frac{1}{\widehat{d}_j}\right)\cdot\min_i(2\widehat{d}^2_i\widehat{c}_{ii}-1)\]
\end{proof}
\end{document}